\documentclass[11pt, draftclsnofoot, onecolumn]{IEEEtran}

\usepackage{graphicx}
\usepackage[usenames,dvipsnames]{xcolor}

\usepackage{amsmath,amssymb,amsfonts,graphicx,bbold,amsthm,mathtools}
\usepackage{enumitem}
\usepackage{subfigure,epstopdf,algorithm}
\usepackage{array}
\usepackage{xcolor}
\usepackage{thmtools}

\declaretheorem{proposition}
\usepackage{times}
\usepackage{eqparbox}
\newcommand{\tr}{\mbox{tr}}
\newcommand{\vecc}{\mbox{vec}}

\usepackage[noend]{algpseudocode}

\makeatletter
\def\BState{\State\hskip-\ALG@thistlm}
\makeatother

\title{Supervised Linear Regression for Graph Learning from Graph Signals}

\author{Arun~Venkitaraman$^*$, Hermina Petric Maretic$^\dagger$, Saikat Chatterjee$^*$, and Pascal Frossard$^\dagger$\\
	$^*$ Department of Information Science and Engineering,\\
	School of Electrical Engineering and Computer Science,     
	KTH Royal Institute of Technology,  Sweden                 \\
	$^\dagger$ Signal Processing Laboratory (LTS4), Ecole Polytechnique F{\'e}d{\'e}rale de Lausanne (EPFL)\\
	arunv@kth.se,	hermina.petric.maretic@epfl.ch, sach@kth.se, pascal.frossard@epfl.ch
}

\begin{document}

	\maketitle
	\begin{abstract}
		We propose a supervised learning approach for predicting an underlying graph from a set of graph signals. Our approach is based on linear regression. In the linear regression model, we predict edge-weights of a graph as the output, given a set of signal values on nodes of the graph as the input. We solve for the optimal regression coefficients using a relevant optimization problem that is convex and uses a graph-Laplacian based regularization. The regularization helps to promote a specific graph spectral profile of the graph signals. Simulation experiments demonstrate that our approach predicts well even in presence of outliers in input data.
	\end{abstract}
	

	\section{Introduction}
	Graph learning in the context of graph signal processing refers to the problem of learning associations between different nodes/agents of a graph or network. A network structure is inferred from the given signal values at the different nodes. 
	Graph learning is part of many analysis and processing tasks, such as clustering, community detection, prediction of signal values, or for predicting entire graph signals. Various models have been proposed to infer a graph from a set of signals \cite{dong2018learning}. Most notable works include graph inference from smooth signals \cite{Dong14} \cite{kalofolias2016learn} \cite{egilmez2017graph}, based on the assumption that signals vary slowly over the graph structure. Pasdeloup et al. \cite{pasdeloup2017characterization} and Segarra et al.  \cite{segarra2017network} assume signals are given as a result of an arbitrary graph filtering process while learning the graph. Similarly, Mei et al. \cite{Mei2017signal} propose a polynomial autoregressive model for graph signals and a method to infer both the graph and coefficients of the polynomial. {\color{black}We note that the aforementioned works take a one-shot approach by learning the graph that best describes a given set of graph signals under suitable constraints. They do not explicitly use a training dataset with labeled graphs and graph signals, and hence, may be seen as unsupervised learning approaches for graph inference.
	
In this paper, we propose a supervised learning approach for predicting graphs from graph signals. A motivating example of supervised graph learning approach can arise in a social network scenario. In social networks, nodes represent different individuals / persons. Let us assume that we have a training dataset comprising graph signals and underlying graphs. The graph signals may comprised of different features, such as age, height, salary, food tastes, consumer habits, etc of the individuals. An underlying graph could be the one formed by a rule based on who follows whom, or a friendship list of individuals. Now, in the case of test data, we may have privacy, security or legal reasons for not revealing the true underlying graph. The task is then to estimate the underlying graph from observed graph signals for the test case. 


To the best of authors' knowledge, there exist no prior work on exploring supervised learning approach for graph learning. A supervised learning approach incorporates prior knowledge through training. 
In our approach, we model the edge-weights of the graph adjacency matrix as the predicted output of a linear regression model with an input consisting of a set of graph signal observations.
	We compute the optimal regression coefficients from training data by solving an optimization problem with a regularization based on the graph spectral profile of the graph signals. In order to make that the optimization problem convex, we use graph spectral profiles in the form of second order polynomial of the graph-Laplacian matrix. We then discuss how for a suitably constructed input feature, the regression coefficients represent a weighting of the different graph signals in the input for the prediction task. Simulation experiments reveal that our approach gives good performance for graph learning under difficult conditions, for examples, if training dataset is limited and noisy, and test input is also noisy. A block scheme summarizing our approach is shown in Figure \ref{fig:schematic}.}

	\begin{figure}
		\centering
		\includegraphics[width=3.6in]{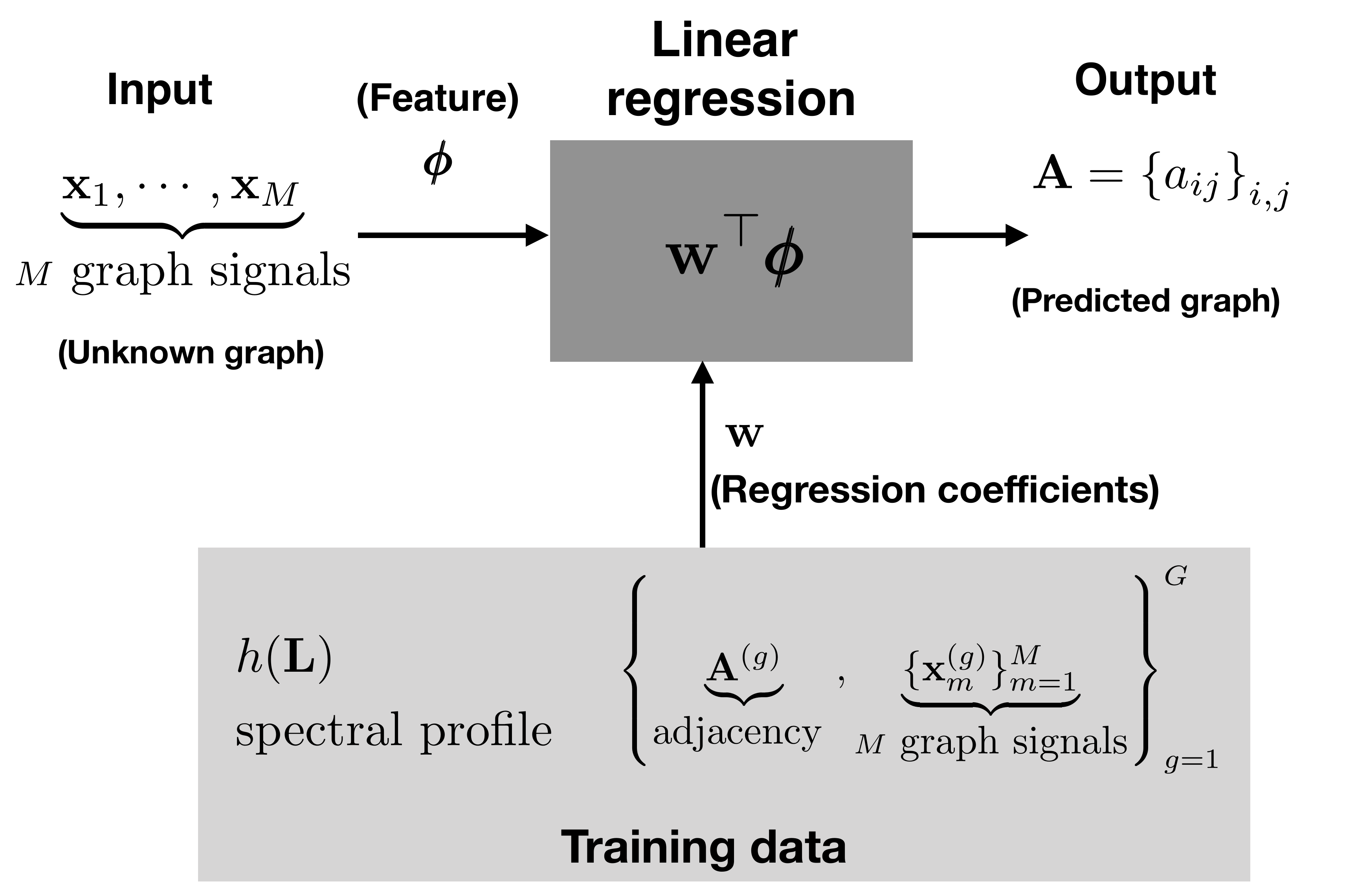}
		\label{fig:schematic}
		\caption{A block scheme of the proposed approach}
	\end{figure}
	\section{Linear regreSsion for graph learning}
	We first review the relevant basics from graph signal processing and thereafter propose linear regression for graph prediction.
	\subsection{Graph signal processing preliminaries}
	\label{sec:gsp}
	A graph signal refers to a vector whose components denote the values of a signal over different nodes of an associated graph. The relation between the different nodes are quantified in using a weighted edge, and the graph is described using the adjacency matrix $\mathbf{A}=[a_{ij}]_{i,j}$ whose $(ij)$th entry $a_{ij}$ denotes the edge-weight between $i$th and $j$th nodes. In this work, we consider only undirected graphs, which means $a_{ij}=a_{ji}$. The smoothness of a graph signal $\mathbf{x}\in\mathbb{R}^N=[x(1),\cdots,x(N)]^\top$ over a graph with $N$ nodes is typically defined using the quantity
	$\mathbf{x}^\top\mathbf{L}\mathbf{x}=\sum_{i,j}a_{ij}(x(i)-x(j))^2,$
	where $\mathbf{L}\triangleq\mathbf{D}-\mathbf{A}$ is the graph-Laplacian matrix\cite{Chung,Shuman}, and $\mathbf{D}=\mbox{diag}(d_1,d_2,\cdots,d_N)=\mbox{diag}(\mathbf{A}\mathbf{1}_N)$ is the diagonal degree matrix with $d_i=\sum_ja_{ij}$, $\mathbf{1}_N$ being the $N$-dimensional vector of all ones. A small value of $\mathbf{x}^\top\mathbf{L}\mathbf{x}$ implies that the values across connected nodes are similar, leading to the notion of a smooth graph signal. A graph signal $\mathbf{x}$ is also equivalently described in terms of its graph-Fourier transform which is defined as
	\begin{equation}
	\hat{\mathbf{x}}\triangleq\mathbf{V}^\top\mathbf{x},\nonumber
	\end{equation}
	where $\mathbf{V}$ denotes the eigenvector matrix of $\mathbf{L}=\mathbf{V}\pmb\Lambda\mathbf{V}^\top$, $\pmb\Lambda=\mbox{diag}(\lambda_1,\lambda_2,\cdots,\lambda_N)$ is the eigenvalue matrix arranged according to ascending values. By construction, $\lambda_1=0$, and the eigenvectors belonging to the smaller $\lambda_i$ vary smoothly over the graph and represent low-frequencies, and those of larger $\lambda_i$ vary more rapidly, denoting the high-frequencies. 
	
	In order to impose that $\mathbf{x}$ follows a particular graph-spectral profile (in terms of the distribution of its graph Fourier spectral coefficients), the regularization $\mathbf{x}^\top h(\mathbf{L})\mathbf{x}$ is often employed, where $h(x)=\sum_{l=0}^{L-1}h_lx^l$ is a polynomial of order $L<N$. This is because the regularization penalizes the different components of $\hat{\mathbf{x}}$ as \begin{equation*}
	\mathbf{x}^\top h(\mathbf{L})\mathbf{x}={\mathbf{x}}^\top\mathbf{V} h(\pmb\Lambda) \mathbf{V}^\top{\mathbf{x}}= \hat{\mathbf{x}}^\top h(\pmb\Lambda) \hat{\mathbf{x}}=\sum_{i=1}^Nh(\lambda_i)\hat{x}(i)^2.
	\end{equation*}
	In the case of smooth graph signals, $\mathbf{H}=\mathbf{L}$ is usually employed since $\mathbf{x}^\top h(\mathbf{L})\mathbf{x}=\hat{\mathbf{x}}^\top \pmb\Lambda\hat{\mathbf{x}}=\sum_{i=2}^N\lambda_i\hat{x}(i)^2$, which penalizes the high-frequency components of $\mathbf{x}$ more than the low-frequency ones. Similarly, setting $\mathbf{H}=\mathbf{L}^\dagger$ where $\mathbf{L}^\dagger$ is the pseudo-inverse of $\mathbf{L}$ leads to $\mathbf{x}^\top h(\mathbf{L})\mathbf{x}=\hat{\mathbf{x}}^\top \pmb\Lambda^\dagger\hat{\mathbf{x}}=\sum_{i=2}^N\frac{1}{\lambda_i}\hat{x}(i)^2$, which promotes $\mathbf{x}$ to have high-frequency behaviour. We refer the reader to \cite{Shuman,gsp_overview_ortega} and the references therein for a more {\color{black}comprehensive} view of graph signal processing framework.
	\subsection{Linear regression model for graph prediction}
	Linear and kernel regression form the workhorse of a gamut of applications which involves learning from from support vector machines \cite{Bishop} to deep learning \cite{kernel_deeplearning} to prediction and reconstruction of graph signals \cite{kergraph1,kergraph4,Arun_kergraph,Arun_mkrg,Arun_GPG,IOANNIDIS2018173}. In this Section, we propose a graph prediction approach using linear regression. We note that in this paper we use the terms graph prediction and graph learning interchangeably.
	
	Let us assume that we have a training set of one or more graphs indexed by $1\leq g\leq G$, $G\geq1$. Let the $g$th graph have $N_g$ nodes. We further assume that we have $M$ graph signals for each of the $G$ graphs as input. Let $\mathbf{A}^{(g)}$ denote the weighted adjacency matrix of the $g$th graph. Then the input-output pairs are given by $\{\mathbf{X}^{(g)},\mathbf{A}^{(g)}\}_{g=1}^G$ where $\mathbf{X}^{(g)}\in\mathbb{R}^{N_g\times M}$ denotes the matrix with the $M$ graph signals as columns. We consider the following model for predicting weight of the edge between the $i$th and $j$th nodes:
	\begin{equation}
	a^{(g)}_{i,j}=\mathbf{w}^\top\pmb\phi\left(\mathbf{x}^{(g)}(i),\mathbf{x}^{(g)}(j)\right) + \textrm{\color{black}model noise}.
	\label{eq:glkr_main}
	\end{equation}
	Here $\mathbf{w}\in\mathbb{R}^K$ is the regression coefficient vector, 
	$\pmb\phi$ is a $K$-dimensional feature vector where $\mathbf{x}^{(g)}(i)$ is the $i$'th row vector of $\mathbf{X}^{(g)}$ as follows
	\begin{eqnarray}
	\mathbf{x}^{(g)}(i)=[x_1^{(g)}(i),\cdots x_m^{(g)}(i),\cdots x_M^{(g)}(i)]^\top\in\mathbb{R}^M.\nonumber
	\end{eqnarray}
	Thus, the estimate of $a^{(g)}_{i,j}$ is given by 
	\begin{align}
	\hat{a}^{(g)}_{i,j}=\mathbf{w}^\top\pmb\phi\left(\mathbf{x}^{(g)}(i),\mathbf{x}^{(g)}(j)\right).
	\label{eq:a_est}
	\end{align}
	{\color{black}The input feature vector $\pmb\phi(\cdot)$ is assumed to be known. In the general case, it could be an arbitrary function of the input signal. Intuitively, for our problem it is desirable that the values of $\pmb\phi(\mathbf{x}(i),\mathbf{x}(j))$ should reflect on the similarity of the signal values between the nodes $i$ and $j$. The smaller the values of 
		$[ (x_1(i)-x_1(j))^2\cdots(x_m(i)-x_m(j))^2\cdots (x_M(i)-x_M(j))^2]$, the larger $\pmb\phi$ must be in order to ensure a strong edge between nodes $i$ and $j$. Similarly, dissimilar values across the nodes with large values of $[ (x_1(i)-x_1(j))^2\cdots (x_M(i)-x_M(j))^2]$ should result in a $\pmb\phi$ with small values. Though multiple such $\pmb\phi$ could be constructed, we use a simple choice with the $m$th component of $\pmb\phi$ defined by
		\begin{align} 
		\pmb\phi\left(\mathbf{x}(i),\mathbf{x}(j)\right)(m)=
		\frac{\sigma}{\max((x_m(i)-x_m(j)^2,\sigma)}.\,\,\,1\leq m\leq M\nonumber\end{align}
		$\sigma$ is a parameter introduced to avoid $\pmb\phi$ being unbounded when the signal values at nodes $i$ and $j$ are very similar. Thus, we observe that the $m$th component of $\pmb\phi$ reflects the similarity of the values of the $m$th graph signal at the $i$'th and $j$'th nodes.
		Correspondingly, the components of $\mathbf{w}$ represent the relative importance of the $M$ graph signals in predicting the graph. In order to ensure that the graphs have no self loops, that is, $a^{(g)}_{ii}=0,\,\,\forall i,g$, we make the additional definition that $\pmb\phi(\mathbf{x}^{(g)}(i),\mathbf{x}^{(g)}(j))
		=\mathbf{0}$ when $i=j$.
		%
		%
		
		Then, by collecting all the edge-weights predicted by the regression model \eqref{eq:a_est} for the $g$th graph as a matrix, we have the adjacency matrix estimate for the $g$th graph given by
		\begin{equation}
		\hat{\mathbf{A}}^{(g)}=\pmb{\Phi}^{(g)}\mathbf{W}_g,\,\,\forall g,\,\,\mbox{where}
		\label{eq:glkr_A}
		\end{equation}
		\begin{align}
		\pmb\Phi^{(g)}&=\left[\begin{matrix}
		\pmb\phi\left(\mathbf{x}^{(g)}(1),\mathbf{x}^{(g)}(1)\right)^\top\, \cdots 
		\pmb\phi\left(\mathbf{x}^{(g)}(1),\mathbf{x}^{(g)}(N_g)\right)^\top\\
		\vdots\\
		\pmb\phi\left(\mathbf{x}^{(g)}(N_g),\mathbf{x}^{(g)}(1)\right)^\top\, \cdots 
		\pmb\phi\left(\mathbf{x}^{(g)}(N_g),\mathbf{x}^{(g)}(N_g)\right)^\top\\
		\end{matrix}\right],\nonumber\\
		\mathbf{W}_g&=\left[\begin{matrix}
		\mathbf{w}\, \mathbf{0}\, \cdots\,\mathbf{0}\\
		\mathbf{0} \,  \mathbf{w}\, \cdots\,\mathbf{0}\\
		\vdots\\
		\mathbf{0}\, \mathbf{0}\,\cdots\,\mathbf{w}
		\end{matrix}\right]=\mathbf{I}_{N_g}\otimes \mathbf{w};\nonumber
		\end{align}
		$\otimes$ denotes the Kronecker product operation and $\mathbf{I}_{N_g}$ is the identity matrix of size $N_g$.}
	Then, the corresponding graph-Laplacian estimate is given by 
	\begin{eqnarray}
	\hat{\mathbf{L}}^{(g)}&=&\hat{\mathbf{D}}^{(g)}-\hat{\mathbf{A}}^{(g)}=\underbrace{\mbox{diag}(\pmb{\Phi}^{(g)}\mathbf{W}_g\mathbf{1}_{N_g})}_{=\mbox{diag}(\mathbf{A}^{(g)}\mathbf{1}_N)=\hat{\mathbf{D}}^{(g)}}-\pmb{\Phi}^{(g)}\mathbf{W}_g,\,\,\forall g\nonumber\\
	&\overset{(a)}{=}&\sum_{n=1}^{N_g}\mathbf{e}_n^\top(\pmb{\Phi}^{(g)}\mathbf{W}_g\mathbf{1}_{N_g})\mathbf{e}_n\mathbf{e}_n^\top-\pmb{\Phi}^{(g)}\mathbf{W}_g,\label{eq:glkr_L}
	\end{eqnarray}
	where $\mathbf{1}_{N_g}$ is the all ones column vector of length $N$ and $\mathbf{e}_n$ is the column vector with all zeros except one at the $n$th component. The equality $(a)$ follows from the matrix identity $\mbox{diag}(\mathbf{a})=\sum_{n=1}^{N}(\mathbf{e}_n^\top\mathbf{a})\mathbf{e}_n\mathbf{e}_n^\top$ where $\mathbf{a}\in\mathbb{R}^N$.
	\begin{figure*}
		\begin{align}
		\label{vecJW}
		\mathbf{g}&=2\vecc\left(\sum_g {\pmb\Phi^{(g)}}^\top\mathbf{A}^{(g)}\right)-\alpha\vecc\left(\sum_{g=1}^G h_1{\pmb{\Phi}^{(g)}}^\top\left[\sum_{n=1}^N\mathbf{e}_n\mathbf{1}_N^\top\tr\left(\mathbf{e}_n\mathbf{e}_n^\top\mathbf{X}^{(g)}{\mathbf{X}^{(g)}}^\top\right)-\mathbf{X}^{(g)}{\mathbf{X}^{(g)}}^\top\right]\right)\nonumber\\
		\mathbf{F}&=
		2\left[\mathbf{I}_N\otimes \sum_g{\pmb\Phi^{(g)}}^\top\pmb\Phi^{(g)}\right]+2\alpha\sum_{g=1}^G h_2\left[\mathbf{X}^{(g)}{\mathbf{X}^{(g)}}^\top\right]^\top\otimes\left[{\pmb{\Phi}^{(g)}}^\top\pmb{\Phi}^{(g)}\right]+\beta\mathbf{I}_{MN}\\
		&\quad+\alpha\sum_{g=1}^G h_2\sum_{n_1=1}^N\sum_{n_2=1}^N\left[(\mathbf{1}_N\mathbf{e}_{n_1}^\top)\otimes{\pmb{\Phi}^{(g)}}^\top\mathbf{1}_N\mathbf{e}_{n_2}^\top\pmb{\Phi}^{(g)}+(\mathbf{e}_{n_1}\mathbf{1}_N^\top)\otimes{\pmb{\Phi}^{(g)}}^\top\mathbf{e}_{n_2}\mathbf{1}^\top_N\pmb{\Phi}^{(g)}\right]\tr\left(\mathbf{e}_{n_1}\mathbf{e}_{n_1}^\top\mathbf{e}_{n_2}\mathbf{e}_{n_2}^\top\mathbf{X}^{(g)}{\mathbf{X}^{(g)}}^\top\right)\nonumber\\
		&\quad+\alpha\sum_{g=1}^G -2h_2\sum_{n=1}^N\vecc\left({\pmb{\Phi}^{(g)}}^\top\mathbf{1}_N\mathbf{e}_n^\top\right)\vecc\left(\left[\mathbf{X}^{(g)}{\mathbf{X}^{(g)}}^\top\mathbf{e}_n\mathbf{e}_n^\top\pmb{\Phi}^{(g)}\right]^\top\right)^\top+\vecc\left({\pmb{\Phi}^{(g)}}^\top\mathbf{e}_n\mathbf{e}_n^\top\mathbf{X}^{(g)}{\mathbf{X}^{(g)}}^\top\right)\vecc\left({\pmb{\Phi}^{(g)}}^\top\mathbf{1}_N\mathbf{e}_n^\top\right)^\top\nonumber\\
		\bar{\pmb\rho}&=[\pmb\rho_1\,\pmb\rho_2\,\cdots\pmb\rho_K]\in\mathbb{R}^{N^2K\times K}\,\,\mbox{where}\,\,	\pmb{\rho}_k=\mathbf{I}_N\otimes\mathbf{e}_k\nonumber\\
		\Omega&=[\Omega_1\, \Omega_2\,\cdots\Omega_N],\,\,\mbox{where   }\,\,\,\Omega_j\triangleq[(j-1)(N+1)K+1:(j-1)(N+1)K+K] 
		\nonumber
		\end{align}
	\end{figure*}
	\subsection{Linear Regression for Graph Prediction}
	Given Eq. \eqref{eq:glkr_main}, \eqref{eq:glkr_A}, \eqref{eq:glkr_L}, our goal is to compute the optimal regression coefficients $\mathbf{w}$ such that the following cost is minimized
	\begin{align}
	J(\mathbf{w})&=
	\sum_{g=1}^G\|\mathbf{A}^{(g)}-\hat{\mathbf{A}}^{(g)}\|_F^2+\alpha\sum_{g=1}^G\tr\left({\mathbf{X}^{(g)}}^\top h(\hat{\mathbf{L}}^{(g)})\mathbf{X}^{(g)}\right)\nonumber\\
	&\quad+\frac{\beta }{G}\sum_{g=1}^G\tr(\mathbf{W}_g^\top\mathbf{W}_g),\nonumber
	\end{align}
	where the first regularization term imposes the learnt graphs to have the desired graph-spectral profile (as discussed in Section \ref{sec:gsp}). The second regularization ensures that $\mathbf{w}$ remains bounded. In imposing the regularization, we have implicitly assumed that the graph signals follow a particular graph spectral profile over the associated graph. 
	This assumption is reasonable in cases such as social networks where the different communities in the graph might still have similar dynamics or distribution of features across the nodes. 
	Now, if we make the further simplifying assumption that all training graphs have the same size $N_g=N$, $J(\mathbf{w})$ is expressible as follows:
	\begin{align}
	\label{eq:glkr_Jw_1}
	J(\mathbf{w})&=\sum_{g=1}^G\|\mathbf{A}^{(g)}-\pmb\Phi^{(g)}\mathbf{W}\|_F^2\\&\quad+\alpha\sum_{g=1}^G\tr\left({\mathbf{X}^{(g)}}^\top h(\hat{\mathbf{L}}^{(g)})\mathbf{X}^{(g)}\right)+\beta \tr(\mathbf{W}^\top\mathbf{W}),\nonumber
	\end{align}
	where $\mathbf{W}=\mathbf{I}_N\otimes\mathbf{w}$. 
	We note that \eqref{eq:glkr_Jw_1} is not convex in $\mathbf{w}$ in general.  Convex optimization problems have a global minimum and often resulting in tractable closed form solutions. This makes it desirable that $J(\mathbf{w})$ in Eq. \eqref{eq:glkr_Jw_1} be convex. This directly translates to the requirement that $h(\cdot)$ be a second order polynomial of the form $h(x)=h_0+h_1x+h_2x^2$. A second order $h(\cdot)$ is nevertheless fairly generic and can represent various kinds of graph signal behaviour such as low-pass, high-pass, etc\cite{Shuman,moura2014bigdata}.
	As $J(\mathbf{w})$ is now convex, the unique global optimal value of $\mathbf{w}$ is obtained by setting the derivative of $J(\mathbf{w})$ with respect to $\mathbf{w}$ equal to zero. This leads to the following proposition:
	\begin{proposition}
		The optimal regression coefficients $\mathbf{w}_{\mbox{opt}}$ that minimizes the cost in Eq. \eqref{eq:glkr_Jw_1} satifies
		$\mathcal{C}_\Omega(\bar{\pmb\rho}^\top\mathbf{F})(\mathbf{1}_N\otimes\mathbf{w}_{\mbox{opt}})=\bar{\pmb\rho}^\top\mathbf{g}$ where $\mathcal{C}_{\Omega}(\mathbf{X})$ denotes the matrix operation that returns the submatrix of $\mathbf{X}$ with only columns indexed by set $\Omega$, and quantities $\Omega$, $\bar{\pmb\rho}$, $\mathbf{F}$, and $\mathbf{g}$ are as defined in Eq. \eqref{vecJW}.
	\end{proposition}
	
	\begin{proof}
		The proof follows from the use of matrix calculus to take the gradient of $J(\mathbf{w})$ with respect to $\mathbf{w}$, and uses chain rule and other standard properties of Kronecker product and vectorization. 
		\begin{align}
		\Aboxed{J(\mathbf{w})=\sum_{g=1}^G\|\mathbf{A}^{(g)}-\pmb\Phi^{(g)}\mathbf{W}\|_F^2+\alpha\sum_{g=1}^G\tr\left({\mathbf{X}^{(g)}}^\top h_g(\mathbf{L}^{(g)})\mathbf{X}^{(g)}\right)+\beta N\tr(\mathbf{w}^\top\mathbf{w})}\nonumber
		\end{align}
		We shall hereafter use $\sum_g$ and $\sum_n$ to denote $\sum_{g=1}^G$ and $\sum_{n=1}^N$, to keep the notation simple. Then, from \eqref{eq:glkr_A} and \eqref{eq:glkr_L} we have that
		\begin{align}
		J(\mathbf{w})&=\underbrace{\sum_{g=1}^G\|\mathbf{A}^{(g)}-\pmb\Phi^{(g)}\mathbf{W}\|_F^2}_{J_1(\mathbf{W})}\nonumber\\&\quad+\alpha\sum_{g=1}^G\underbrace{\tr\left({\mathbf{X}^{(g)}}^\top h_g\left(\sum_{n=1}^N\mathbf{e}_n^\top(\pmb{\Phi}^{(g)}\mathbf{W})(\mathbf{1}_N\mathbf{e}_n)\mathbf{e}_n^\top-\pmb{\Phi}^{(g)}\mathbf{W}\right)\mathbf{X}^{(g)}\right)}_{J_2(\mathbf{W})}\nonumber\\
		&\quad+\beta \underbrace{\tr(\mathbf{W}^\top\mathbf{W})}_{J_3(\mathbf{W})}\nonumber\\
		&=J_1(\mathbf{W})+\alpha J_2(\mathbf{W})+\beta J_3(\mathbf{W})
		\label{eq:glkr_Jw}
		\end{align}
		\subsection{Simplifying cost function}
		We now analyze these terms separately:
		\begin{align}
		J_1(\mathbf{W})&=\sum_g\|\mathbf{A}^{(g)}-\pmb\Phi^{(g)}\mathbf{W}\|_F^2=\sum_g\tr([\mathbf{A}^{(g)}-\pmb\Phi^{(g)}\mathbf{W}]^\top[\mathbf{A}^{(g)}-\pmb\Phi^{(g)}\mathbf{W}])\nonumber\\
		&=\sum_g\tr([\mathbf{A}^{(g)}]^\top\mathbf{A}^{(g)})-2\sum_g\tr([\mathbf{A}^{(g)}]^\top\pmb\Phi^{(g)}\mathbf{W})+\sum_g\tr([\pmb\Phi^{(g)}\mathbf{W}]^\top[\pmb\Phi^{(g)}\mathbf{W}])\nonumber\\
		&=\sum_g\tr([\mathbf{A}^{(g)}]^\top\mathbf{A}^{(g)})-2\sum_g\tr([\mathbf{A}^{(g)}]^\top\pmb\Phi^{(g)}\mathbf{W})+\sum_g\tr(\mathbf{W}^\top{\pmb\Phi^{(g)}}^\top\pmb\Phi^{(g)}\mathbf{W})\nonumber\\
		&=\sum_g\tr([\mathbf{A}^{(g)}]^\top\mathbf{A}^{(g)})-2\tr(\sum_g[\mathbf{A}^{(g)}]^\top\pmb\Phi^{(g)}\mathbf{W})+\tr(\mathbf{W}^\top\sum_g[{\pmb\Phi^{(g)}}^\top\pmb\Phi^{(g)}]\mathbf{W})\nonumber\\
		\label{eq:glkr_J1}
		\end{align}
		\begin{align}
		J_2(\mathbf{W})&=\sum_{g=1}^G\tr\left({\mathbf{X}^{(g)}}^\top h_g\left(\mathbf{L}^{(g)}\right)\mathbf{X}^{(g)}\right)
		=\sum_{g=1}^G\tr\left( h_g\left({\mathbf{L}}^{(g)}\right)\mathbf{X}^{(g)}{\mathbf{X}^{(g)}}^\top\right)\nonumber\\
		&=\sum_{g=1}^G\tr\left( h_0^{(g)}\mathbf{X}^{(g)}{\mathbf{X}^{(g)}}^\top\right)+\sum_{g=1}^G\tr\left( h_1^{(g)}{\mathbf{L}}^{(g)}\mathbf{X}^{(g)}{\mathbf{X}^{(g)}}^\top\right)+\sum_{g=1}^G\tr\left( h_2^{(g)}{\mathbf{L}^{(g)}}^2\mathbf{X}^{(g)}{\mathbf{X}^{(g)}}^\top\right)\nonumber\\
		&=\sum_{g=1}^G\tr\left( h_0^{(g)}\mathbf{X}^{(g)}{\mathbf{X}^{(g)}}^\top\right)\nonumber\\
		&\quad+\sum_{g=1}^G\tr\left( h_1^{(g)}\left[\sum_{n=1}^N\underbrace{\mathbf{e}_n^\top(\pmb{\Phi}^{(g)}\mathbf{W}\mathbf{1}_N)}_{\mbox{scalar}}\mathbf{e}_n\mathbf{e}_n^\top-\pmb{\Phi}^{(g)}\mathbf{W}\right]\mathbf{X}^{(g)}{\mathbf{X}^{(g)}}^\top\right)\nonumber\\
		&\quad+\sum_{g=1}^G\tr\left( h_2^{(g)}\left[\sum_{n=1}^N\mathbf{e}_n^\top(\pmb{\Phi}^{(g)}\mathbf{W}\mathbf{1}_N)\mathbf{e}_n\mathbf{e}_n^\top-\pmb{\Phi}^{(g)}\mathbf{W}\right]^2\mathbf{X}^{(g)}{\mathbf{X}^{(g)}}^\top\right)\nonumber\\
		&=\sum_{g=1}^G\tr\left( h_0^{(g)}\mathbf{X}^{(g)}{\mathbf{X}^{(g)}}^\top\right)\nonumber\\
		&\quad+\sum_{g=1}^G\tr\left( h_1^{(g)}\left[\sum_{n=1}^N\mathbf{e}_n^\top(\pmb{\Phi}^{(g)}\mathbf{W}\mathbf{1}_N)\mathbf{e}_n\mathbf{e}_n^\top-\pmb{\Phi}^{(g)}\mathbf{W}\right]\mathbf{X}^{(g)}{\mathbf{X}^{(g)}}^\top\right)\nonumber\\
		&\quad+\sum_{g=1}^G\tr\left( h_2^{(g)}\left[\sum_{n_1=1}^N\sum_{n_2=1}^N\mathbf{e}_{n_1}^\top(\pmb{\Phi}^{(g)}\mathbf{W}\mathbf{1}_N)\mathbf{e}_{n_1}\mathbf{e}_{n_1}^\top\mathbf{e}_{n_2}^\top(\pmb{\Phi}^{(g)}\mathbf{W}\mathbf{1}_N)\mathbf{e}_{n_2}\mathbf{e}_{n_2}^\top\right]\mathbf{X}^{(g)}{\mathbf{X}^{(g)}}^\top\right)\nonumber\\
		&\quad+\sum_{g=1}^G\tr\left( -2h_2^{(g)}\left[\sum_{n=1}^N\mathbf{e}_n^\top(\pmb{\Phi}^{(g)}\mathbf{W}\mathbf{1}_N)\mathbf{e}_n\mathbf{e}_n^\top\pmb{\Phi}^{(g)}\mathbf{W}\right]\mathbf{X}^{(g)}{\mathbf{X}^{(g)}}^\top\right)\nonumber\\
		&\quad+\sum_{g=1}^G\tr\left( h_2^{(g)}\left[\pmb{\Phi}^{(g)}\mathbf{W}\pmb{\Phi}^{(g)}\mathbf{W}\right]\mathbf{X}^{(g)}{\mathbf{X}^{(g)}}^\top\right)\nonumber\\
		\end{align}
		\begin{align}
		J_2(\mathbf{W})&\overset{(a)}{=}\sum_{g=1}^G\tr\left( h_0^{(g)}\mathbf{X}^{(g)}{\mathbf{X}^{(g)}}^\top\right)\nonumber\\
		&\quad+\sum_{g=1}^G\tr\left( h_1^{(g)}\left[\sum_{n=1}^N\mathbf{e}_n^\top(\pmb{\Phi}^{(g)}\mathbf{W}\mathbf{1}_N)\mathbf{e}_n\mathbf{e}_n^\top\right]\mathbf{X}^{(g)}{\mathbf{X}^{(g)}}^\top\right)\nonumber\\
		&\quad+\sum_{g=1}^G\tr\left( h_1^{(g)}\left[-\pmb{\Phi}^{(g)}\mathbf{W}\right]\mathbf{X}^{(g)}{\mathbf{X}^{(g)}}^\top\right)\nonumber\\
		&\quad+\sum_{g=1}^G\tr\left( h_2^{(g)}\left[\sum_{n_1=1}^N\sum_{n_2=1}^N\mathbf{e}_{n_1}^\top(\mathbf{W}^\top{\pmb{\Phi}^{(g)}}^\top\mathbf{1}_N)\mathbf{e}_{n_1}\mathbf{e}_{n_1}^\top\mathbf{e}_{n_2}^\top(\pmb{\Phi}^{(g)}\mathbf{W}\mathbf{1}_N)\mathbf{e}_{n_2}\mathbf{e}_{n_2}^\top\right]\mathbf{X}^{(g)}{\mathbf{X}^{(g)}}^\top\right)\nonumber\\
		&\quad+\sum_{g=1}^G\tr\left( -2h_2^{(g)}\left[\sum_{n=1}^N\mathbf{e}_n^\top(\mathbf{W}^\top{\pmb{\Phi}^{(g)}}^\top\mathbf{1}_N)\mathbf{e}_n\mathbf{e}_n^\top\pmb{\Phi}^{(g)}\mathbf{W}\right]\mathbf{X}^{(g)}{\mathbf{X}^{(g)}}^\top\right)\nonumber\\
		&\quad+\sum_{g=1}^G\tr\left( h_2^{(g)}\left[\mathbf{W}^\top{\pmb{\Phi}^{(g)}}^\top\pmb{\Phi}^{(g)}\mathbf{W}\right]\mathbf{X}^{(g)}{\mathbf{X}^{(g)}}^\top\right)\nonumber
		\end{align}
		\begin{align}
		J_2(\mathbf{W})&\overset{(a)}{=}\sum_{g=1}^G\tr\left( h_0^{(g)}\mathbf{X}^{(g)}{\mathbf{X}^{(g)}}^\top\right)\nonumber\\
		&\quad+\sum_{g=1}^G h_1^{(g)}\sum_{n=1}^N\mathbf{e}_n^\top(\pmb{\Phi}^{(g)}\mathbf{W}\mathbf{1}_N)\tr\left(\mathbf{e}_n\mathbf{e}_n^\top\mathbf{X}^{(g)}{\mathbf{X}^{(g)}}^\top\right)\nonumber\\
		&\quad+\sum_{g=1}^G\tr\left( h_1^{(g)}\left[-\pmb{\Phi}^{(g)}\mathbf{W}\right]\mathbf{X}^{(g)}{\mathbf{X}^{(g)}}^\top\right)\nonumber\\
		&\quad+\sum_{g=1}^G h_2^{(g)}\sum_{n_1=1}^N\sum_{n_2=1}^N\mathbf{e}_{n_1}^\top(\mathbf{W}^\top{\pmb{\Phi}^{(g)}}^\top\mathbf{1}_N)\mathbf{e}_{n_2}^\top(\pmb{\Phi}^{(g)}\mathbf{W}\mathbf{1}_N)\tr\left(\mathbf{e}_{n_1}\mathbf{e}_{n_1}^\top\mathbf{e}_{n_2}\mathbf{e}_{n_2}^\top\mathbf{X}^{(g)}{\mathbf{X}^{(g)}}^\top\right)\nonumber\\
		&\quad+\sum_{g=1}^G -2h_2^{(g)}\sum_{n=1}^N\mathbf{e}_n^\top(\mathbf{W}^\top{\pmb{\Phi}^{(g)}}^\top\mathbf{1}_N)\tr\left(\mathbf{e}_n\mathbf{e}_n^\top\pmb{\Phi}^{(g)}\mathbf{W}\mathbf{X}^{(g)}{\mathbf{X}^{(g)}}^\top\right)\nonumber\\
		&\quad+\sum_{g=1}^G\tr\left( h_2^{(g)}\left[\mathbf{W}^\top{\pmb{\Phi}^{(g)}}^\top\pmb{\Phi}^{(g)}\mathbf{W}\right]\mathbf{X}^{(g)}{\mathbf{X}^{(g)}}^\top\right)\nonumber
		\end{align}
		
		\begin{align}
		J_2(\mathbf{W})&\overset{(a)}{=}\sum_{g=1}^G\tr\left( h_0^{(g)}\mathbf{X}^{(g)}{\mathbf{X}^{(g)}}^\top\right)\nonumber\\
		&\quad+\sum_{g=1}^G h_1^{(g)}\sum_{n=1}^N\tr\left(\mathbf{e}_n^\top(\pmb{\Phi}^{(g)}\mathbf{W}\mathbf{1}_N)\right)\tr\left(\mathbf{e}_n\mathbf{e}_n^\top\mathbf{X}^{(g)}{\mathbf{X}^{(g)}}^\top\right)\nonumber\\
		&\quad+\sum_{g=1}^G\tr\left( h_1^{(g)}\left[-\pmb{\Phi}^{(g)}\mathbf{W}\right]\mathbf{X}^{(g)}{\mathbf{X}^{(g)}}^\top\right)\nonumber\\
		&\quad+\sum_{g=1}^G h_2^{(g)}\sum_{n_1=1}^N\sum_{n_2=1}^N\tr\left(\mathbf{e}_{n_1}^\top(\mathbf{W}^\top{\pmb{\Phi}^{(g)}}^\top\mathbf{1}_N)\mathbf{e}_{n_2}^\top(\pmb{\Phi}^{(g)}\mathbf{W}\mathbf{1}_N)\right)\tr\left(\mathbf{e}_{n_1}\mathbf{e}_{n_1}^\top\mathbf{e}_{n_2}\mathbf{e}_{n_2}^\top\mathbf{X}^{(g)}{\mathbf{X}^{(g)}}^\top\right)\nonumber\\
		&\quad+\sum_{g=1}^G -2h_2^{(g)}\sum_{n=1}^N\tr\left(\mathbf{e}_n^\top(\mathbf{W}^\top{\pmb{\Phi}^{(g)}}^\top\mathbf{1}_N)\right)\tr\left(\mathbf{X}^{(g)}{\mathbf{X}^{(g)}}^\top\mathbf{e}_n\mathbf{e}_n^\top\pmb{\Phi}^{(g)}\mathbf{W}\right)\nonumber\\
		&\quad+\sum_{g=1}^G\tr\left( h_2^{(g)}\left[\mathbf{W}^\top{\pmb{\Phi}^{(g)}}^\top\pmb{\Phi}^{(g)}\mathbf{W}\right]\mathbf{X}^{(g)}{\mathbf{X}^{(g)}}^\top\right)\nonumber
		\end{align}
		
		\begin{align}
		J_2(\mathbf{W})&\overset{(a)}{=}\sum_{g=1}^G\tr\left( h_0^{(g)}\mathbf{X}^{(g)}{\mathbf{X}^{(g)}}^\top\right)\nonumber\\
		&\quad+\sum_{g=1}^G h_1^{(g)}\sum_{n=1}^N\tr\left(\mathbf{W}\mathbf{1}_N\mathbf{e}_n^\top \pmb{\Phi}^{(g)}\right)\tr\left(\mathbf{e}_n\mathbf{e}_n^\top\mathbf{X}^{(g)}{\mathbf{X}^{(g)}}^\top\right)\nonumber\\
		&\quad+\sum_{g=1}^G\tr\left( h_1^{(g)}\left[-\pmb{\Phi}^{(g)}\mathbf{W}\right]\mathbf{X}^{(g)}{\mathbf{X}^{(g)}}^\top\right)\nonumber\\
		&\quad+\sum_{g=1}^G h_2^{(g)}\sum_{n_1=1}^N\sum_{n_2=1}^N\tr\left(\mathbf{W}^\top{\pmb{\Phi}^{(g)}}^\top\mathbf{1}_N\mathbf{e}_{n_2}^\top\pmb{\Phi}^{(g)}\mathbf{W}\mathbf{1}_N\mathbf{e}_{n_1}^\top\right)\tr\left(\mathbf{e}_{n_1}\mathbf{e}_{n_1}^\top\mathbf{e}_{n_2}\mathbf{e}_{n_2}^\top\mathbf{X}^{(g)}{\mathbf{X}^{(g)}}^\top\right)\nonumber\\
		&\quad+\sum_{g=1}^G -2h_2^{(g)}\sum_{n=1}^N\tr\left(\mathbf{W}^\top{\pmb{\Phi}^{(g)}}^\top\mathbf{1}_N\mathbf{e}_n^\top\right)\tr\left(\mathbf{X}^{(g)}{\mathbf{X}^{(g)}}^\top\mathbf{e}_n\mathbf{e}_n^\top\pmb{\Phi}^{(g)}\mathbf{W}\right)\nonumber\\
		&\quad+\sum_{g=1}^G\tr\left( h_2^{(g)}\left[\mathbf{W}^\top{\pmb{\Phi}^{(g)}}^\top\pmb{\Phi}^{(g)}\mathbf{W}\right]\mathbf{X}^{(g)}{\mathbf{X}^{(g)}}^\top\right)
		\label{eq:glkr_J2_2}
		\end{align}

		\subsection{Taking derivatives of the cost function parts with respect to $\mathbf{W}$}
		In order to keep the mathematics self-contained, we list here some properties of matrix calculus which we shall be using later:
		\begin{align}
		\vecc(\mathbf{AXB})=(\mathbf{B}^\top\otimes\mathbf{A})\vecc(\mathbf{X})\nonumber\\
		\frac{\partial \tr(\mathbf{AX})}{\partial\mathbf{X}}=\mathbf{A}^\top\nonumber\\
		\frac{\partial \tr(\mathbf{X^\top AXB})}{\partial\mathbf{X}}=\mathbf{A}\mathbf{X}\mathbf{B}+\mathbf{A}^\top\mathbf{X}\mathbf{B}^\top\nonumber\\
		\tr{(\mathbf{A}^\top\mathbf{B})}=(\vecc\mathbf{A})^\top\vecc\mathbf{B}\nonumber
		\end{align}
		Then, from \eqref{eq:glkr_J1} we have that
		\begin{align}
		\frac{\partial J_1(\mathbf{W})}{\partial\mathbf{W}}&=-2\sum_g {\pmb\Phi^{(g)}}^\top\mathbf{A}^{(g)}+2\sum_g[{\pmb\Phi^{(g)}}^\top\pmb\Phi^{(g)}]\mathbf{W}
		\label{eq:glkr_deriv_J1}
		\end{align}
		Similarly, from \eqref{eq:glkr_J2_2} we have that
		\begin{align}
		\frac{\partial J_2(\mathbf{W})}{\partial\mathbf{W}}&\overset{(a)}{=}\sum_{g=1}^G h_1^{(g)}\sum_{n=1}^N{\pmb{\Phi}^{(g)}}^\top\mathbf{e}_n\mathbf{1}_N^\top\tr\left(\mathbf{e}_n\mathbf{e}_n^\top\mathbf{X}^{(g)}{\mathbf{X}^{(g)}}^\top\right)\nonumber\\
		&\quad-\sum_{g=1}^G h_1^{(g)}{\pmb{\Phi}^{(g)}}^\top\mathbf{X}^{(g)}{\mathbf{X}^{(g)}}^\top\nonumber\\
		&\quad+\sum_{g=1}^G h_2^{(g)}\sum_{n_1=1}^N\sum_{n_2=1}^N{\pmb{\Phi}^{(g)}}^\top\mathbf{1}_N\mathbf{e}_{n_2}^\top\pmb{\Phi}^{(g)}\mathbf{W}\mathbf{1}_N\mathbf{e}_{n_1}^\top\tr\left(\mathbf{e}_{n_1}\mathbf{e}_{n_1}^\top\mathbf{e}_{n_2}\mathbf{e}_{n_2}^\top\mathbf{X}^{(g)}{\mathbf{X}^{(g)}}^\top\right)\nonumber\\
		&\quad+\sum_{g=1}^G h_2^{(g)}\sum_{n_1=1}^N\sum_{n_2=1}^N{\pmb{\Phi}^{(g)}}^\top\mathbf{e}_{n_2}\mathbf{1}^\top_N\pmb{\Phi}^{(g)}\mathbf{W}\mathbf{e}_{n_1}\mathbf{1}_N^\top\tr\left(\mathbf{e}_{n_1}\mathbf{e}_{n_1}^\top\mathbf{e}_{n_2}\mathbf{e}_{n_2}^\top\mathbf{X}^{(g)}{\mathbf{X}^{(g)}}^\top\right)\nonumber\\
		&\quad+\sum_{g=1}^G -2h_2^{(g)}\sum_{n=1}^N{\pmb{\Phi}^{(g)}}^\top\mathbf{1}_N\mathbf{e}_n^\top\tr\left(\mathbf{X}^{(g)}{\mathbf{X}^{(g)}}^\top\mathbf{e}_n\mathbf{e}_n^\top\pmb{\Phi}^{(g)}\mathbf{W}\right)\nonumber\\
		&\quad+\sum_{g=1}^G -2h_2^{(g)}\sum_{n=1}^N\tr\left(\mathbf{W}^\top{\pmb{\Phi}^{(g)}}^\top\mathbf{1}_N\mathbf{e}_n^\top\right){\pmb{\Phi}^{(g)}}^\top\mathbf{e}_n\mathbf{e}_n^\top\mathbf{X}^{(g)}{\mathbf{X}^{(g)}}^\top \nonumber\\
		&\quad+2\sum_{g=1}^G h_2^{(g)}\left[{\pmb{\Phi}^{(g)}}^\top\pmb{\Phi}^{(g)}\right]\mathbf{W}\left[\mathbf{X}^{(g)}{\mathbf{X}^{(g)}}^\top\right]
		\label{eq:glkr_deriv_J2}
		\end{align}
		And finally,
		\begin{align}
		\frac{\partial J_3(\mathbf{W})}{\partial\mathbf{W}}=2\mathbf{W}
		\label{eq:glkr_deriv_J3}
		\end{align}
		Then, from \eqref{eq:glkr_deriv_J1}, \eqref{eq:glkr_deriv_J2}, and \eqref{eq:glkr_deriv_J3}, we have that
		\begin{align}
		\frac{\partial J(\mathbf{W})}{\partial\mathbf{W}}&=-2\sum_g {\pmb\Phi^{(g)}}^\top\mathbf{A}^{(g)}+2\sum_g[{\pmb\Phi^{(g)}}^\top\pmb\Phi^{(g)}]\mathbf{W}\nonumber\\
		&\quad+\alpha\sum_{g=1}^G h_1^{(g)}\sum_{n=1}^N{\pmb{\Phi}^{(g)}}^\top\mathbf{e}_n\mathbf{1}_N^\top\tr\left(\mathbf{e}_n\mathbf{e}_n^\top\mathbf{X}^{(g)}{\mathbf{X}^{(g)}}^\top\right)\nonumber\\
		&\quad-\alpha\sum_{g=1}^G h_1^{(g)}{\pmb{\Phi}^{(g)}}^\top\mathbf{X}^{(g)}{\mathbf{X}^{(g)}}^\top\nonumber\\
		&\quad+\alpha\sum_{g=1}^G h_2^{(g)}\sum_{n_1=1}^N\sum_{n_2=1}^N{\pmb{\Phi}^{(g)}}^\top\mathbf{1}_N\mathbf{e}_{n_2}^\top\pmb{\Phi}^{(g)}\mathbf{W}\mathbf{1}_N\mathbf{e}_{n_1}^\top\tr\left(\mathbf{e}_{n_1}\mathbf{e}_{n_1}^\top\mathbf{e}_{n_2}\mathbf{e}_{n_2}^\top\mathbf{X}^{(g)}{\mathbf{X}^{(g)}}^\top\right)\nonumber\\
		&\quad+\alpha\sum_{g=1}^G h_2^{(g)}\sum_{n_1=1}^N\sum_{n_2=1}^N{\pmb{\Phi}^{(g)}}^\top\mathbf{e}_{n_2}\mathbf{1}^\top_N\pmb{\Phi}^{(g)}\mathbf{W}\mathbf{e}_{n_1}\mathbf{1}_N^\top\tr\left(\mathbf{e}_{n_1}\mathbf{e}_{n_1}^\top\mathbf{e}_{n_2}\mathbf{e}_{n_2}^\top\mathbf{X}^{(g)}{\mathbf{X}^{(g)}}^\top\right)\nonumber\\
		&\quad+\alpha\sum_{g=1}^G -2h_2^{(g)}\sum_{n=1}^N{\pmb{\Phi}^{(g)}}^\top\mathbf{1}_N\mathbf{e}_n^\top\tr\left(\mathbf{X}^{(g)}{\mathbf{X}^{(g)}}^\top\mathbf{e}_n\mathbf{e}_n^\top\pmb{\Phi}^{(g)}\mathbf{W}\right)\nonumber\\
		&\quad+\alpha\sum_{g=1}^G -2h_2^{(g)}\sum_{n=1}^N\tr\left(\mathbf{W}^\top{\pmb{\Phi}^{(g)}}^\top\mathbf{1}_N\mathbf{e}_n^\top\right){\pmb{\Phi}^{(g)}}^\top\mathbf{e}_n\mathbf{e}_n^\top\mathbf{X}^{(g)}{\mathbf{X}^{(g)}}^\top \nonumber\\
		&\quad+2\alpha\sum_{g=1}^G h_2^{(g)}\left[{\pmb{\Phi}^{(g)}}^\top\pmb{\Phi}^{(g)}\right]\mathbf{W}\left[\mathbf{X}^{(g)}{\mathbf{X}^{(g)}}^\top\right]
		\nonumber\\
		&\quad+2\beta\mathbf{W}
		\label{eq:glkr_deriv_J}
		\end{align}
		\subsection{Vectorizing the derivatives}
		In our later computations, we would need the vectorized version of these derivatives, which we now compute:
		
		\begin{align}
		\vecc\left(\frac{\partial J_1(\mathbf{W})}{\partial\mathbf{W}}\right)&=-2\vecc\left(\sum_g {\pmb\Phi^{(g)}}^\top\mathbf{A}^{(g)}\right)+2\left[\mathbf{I}_N\otimes \sum_g{\pmb\Phi^{(g)}}^\top\pmb\Phi^{(g)}\right]\vecc\left(\mathbf{W}\right)\nonumber
		\end{align}
		Similarly, we have
		\begin{align}
		\vecc\left(\frac{\partial J_2(\mathbf{W})}{\partial\mathbf{W}}\right)&=\vecc\left(\sum_{g=1}^G h_1^{(g)}{\pmb{\Phi}^{(g)}}^\top\left[\sum_{n=1}^N\mathbf{e}_n\mathbf{1}_N^\top\tr\left(\mathbf{e}_n\mathbf{e}_n^\top\mathbf{X}^{(g)}{\mathbf{X}^{(g)}}^\top\right)-\mathbf{X}^{(g)}{\mathbf{X}^{(g)}}^\top\right]\right)\nonumber\\
		&\quad+\sum_{g=1}^G h_2^{(g)}\sum_{n_1=1}^N\sum_{n_2=1}^N\left[(\mathbf{1}_N\mathbf{e}_{n_1}^\top)\otimes{\pmb{\Phi}^{(g)}}^\top\mathbf{1}_N\mathbf{e}_{n_2}^\top\pmb{\Phi}^{(g)}\right]\tr\left(\mathbf{e}_{n_1}\mathbf{e}_{n_1}^\top\mathbf{e}_{n_2}\mathbf{e}_{n_2}^\top\mathbf{X}^{(g)}{\mathbf{X}^{(g)}}^\top\right)\vecc\mathbf{W}\nonumber\\
		&\quad+\sum_{g=1}^G h_2^{(g)}\sum_{n_1=1}^N\sum_{n_2=1}^N\left[(\mathbf{e}_{n_1}\mathbf{1}_N^\top)\otimes{\pmb{\Phi}^{(g)}}^\top\mathbf{e}_{n_2}\mathbf{1}^\top_N\pmb{\Phi}^{(g)}\right]\tr\left(\mathbf{e}_{n_1}\mathbf{e}_{n_1}^\top\mathbf{e}_{n_2}\mathbf{e}_{n_2}^\top\mathbf{X}^{(g)}{\mathbf{X}^{(g)}}^\top\right)\vecc\mathbf{W}\nonumber\\
		&\quad+\sum_{g=1}^G -2h_2^{(g)}\sum_{n=1}^N\vecc\left({\pmb{\Phi}^{(g)}}^\top\mathbf{1}_N\mathbf{e}_n^\top\right)\vecc\left(\left[\mathbf{X}^{(g)}{\mathbf{X}^{(g)}}^\top\mathbf{e}_n\mathbf{e}_n^\top\pmb{\Phi}^{(g)}\right]^\top\right)^\top\vecc\mathbf{W}\nonumber\\
		&\quad+\sum_{g=1}^G -2h_2^{(g)}\sum_{n=1}^N\vecc\left({\pmb{\Phi}^{(g)}}^\top\mathbf{e}_n^\top\mathbf{e}_n\mathbf{X}^{(g)}{\mathbf{X}^{(g)}}^\top\right)\vecc\left({\pmb{\Phi}^{(g)}}^\top\mathbf{1}_N\mathbf{e}_n^\top\right)^\top\vecc\mathbf{W} \nonumber\\
		&\quad+2\sum_{g=1}^G h_2^{(g)}\left[\mathbf{X}^{(g)}{\mathbf{X}^{(g)}}^\top\right]^\top\otimes\left[{\pmb{\Phi}^{(g)}}^\top\pmb{\Phi}^{(g)}\right]\vecc\mathbf{W}\nonumber
		\end{align}
		Putting together vectorized partial derivatives of $J_1$, $J_2$, and $J_3$, we have that the vectorized partial derivative of $J$ with respect to $\mathbf{W}$ is given by:
		\begin{align}
		\vecc\left(\frac{\partial J(\mathbf{W})}{\partial\mathbf{W}}\right)&=-2\vecc\left(\sum_g {\pmb\Phi^{(g)}}^\top\mathbf{A}^{(g)}\right)+2\left[\mathbf{I}_N\otimes \sum_g{\pmb\Phi^{(g)}}^\top\pmb\Phi^{(g)}\right]\vecc\left(\mathbf{W}\right)\nonumber\\
		&\quad+\alpha\vecc\left(\sum_{g=1}^G h_1^{(g)}{\pmb{\Phi}^{(g)}}^\top\left[\sum_{n=1}^N\mathbf{e}_n\mathbf{1}_N^\top\tr\left(\mathbf{e}_n\mathbf{e}_n^\top\mathbf{X}^{(g)}{\mathbf{X}^{(g)}}^\top\right)-\mathbf{X}^{(g)}{\mathbf{X}^{(g)}}^\top\right]\right)\nonumber\\
		&\quad+\alpha\sum_{g=1}^G h_2^{(g)}\sum_{n_1=1}^N\sum_{n_2=1}^N\left[(\mathbf{1}_N\mathbf{e}_{n_1}^\top)\otimes{\pmb{\Phi}^{(g)}}^\top\mathbf{1}_N\mathbf{e}_{n_2}^\top\pmb{\Phi}^{(g)}\right]\tr\left(\mathbf{e}_{n_1}\mathbf{e}_{n_1}^\top\mathbf{e}_{n_2}\mathbf{e}_{n_2}^\top\mathbf{X}^{(g)}{\mathbf{X}^{(g)}}^\top\right)\vecc\mathbf{W}\nonumber\\
		&\quad+\alpha\sum_{g=1}^G h_2^{(g)}\sum_{n_1=1}^N\sum_{n_2=1}^N\left[(\mathbf{e}_{n_1}\mathbf{1}_N^\top)\otimes{\pmb{\Phi}^{(g)}}^\top\mathbf{e}_{n_2}\mathbf{1}^\top_N\pmb{\Phi}^{(g)}\right]\tr\left(\mathbf{e}_{n_1}\mathbf{e}_{n_1}^\top\mathbf{e}_{n_2}\mathbf{e}_{n_2}^\top\mathbf{X}^{(g)}{\mathbf{X}^{(g)}}^\top\right)\vecc\mathbf{W}\nonumber\\
		&\quad+\alpha\sum_{g=1}^G -2h_2^{(g)}\sum_{n=1}^N\vecc\left({\pmb{\Phi}^{(g)}}^\top\mathbf{1}_N\mathbf{e}_n^\top\right)\vecc\left(\left[\mathbf{X}^{(g)}{\mathbf{X}^{(g)}}^\top\mathbf{e}_n\mathbf{e}_n^\top\pmb{\Phi}^{(g)}\right]^\top\right)^\top\vecc\mathbf{W}\nonumber\\
		&\quad+\alpha\sum_{g=1}^G -2h_2^{(g)}\sum_{n=1}^N\vecc\left({\pmb{\Phi}^{(g)}}^\top\mathbf{e}_n\mathbf{e}_n^\top\mathbf{X}^{(g)}{\mathbf{X}^{(g)}}^\top\right)\vecc\left({\pmb{\Phi}^{(g)}}^\top\mathbf{1}_N\mathbf{e}_n^\top\right)^\top\vecc\mathbf{W} \nonumber\\
		&\quad+2\alpha\sum_{g=1}^G h_2^{(g)}\left[\mathbf{X}^{(g)}{\mathbf{X}^{(g)}}^\top\right]^\top\otimes\left[{\pmb{\Phi}^{(g)}}^\top\pmb{\Phi}^{(g)}\right]\vecc\mathbf{W}\nonumber\\
		&\quad+\beta\vecc\mathbf{W}\nonumber\\
		&=\mathbf{F}\vecc\mathbf{W}-\mathbf{g}
		\end{align}
		The derivative of $J(\mathbf{W})$ with respect to $k$ th component of $\mathbf{w}$  denoted by $\mathbf{w}(k)$ is then given by
		\begin{align}
		\frac{\partial J(\mathbf{W})}{\partial{\mathbf{w}(k)}}=\tr\left(\left[\frac{\partial J(\mathbf{W})}{\partial\mathbf{W}}\right]^\top\frac{\partial \mathbf{W}}{\partial{\mathbf{w}(k)}}\right)=\left(\vecc\frac{\partial J(\mathbf{W})}{\partial\mathbf{W}}\right)^\top\vecc\frac{\partial \mathbf{W}}{\partial{\mathbf{w}(k)}}\nonumber
		\end{align}
		Now, we have that
		\begin{align}
		\frac{\partial \mathbf{W}}{\partial{\mathbf{w}(k)}}=\frac{\partial (\mathbf{I}_N\otimes\mathbf{w})}{\partial{\mathbf{w}(k)}}=\mathbf{I}_N\otimes\frac{\partial \mathbf{w}}{\partial{\mathbf{w}(k)}}=\mathbf{I}_N\otimes\mathbf{e}_k\nonumber
		\end{align}
		Let $\mathbf{\rho}_k\in\mathbb{R}^{N^2K}$ denote $\vecc\frac{\partial \mathbf{W}}{\partial{\mathbf{w}(k)}}$. Then, we have that
		\begin{align}
		\mathbf{\rho}_k(i)=\begin{cases}
		1, & \mbox{if}\,\, i=m((N+1)K)+k, \,\, m=0,1,\cdots,N-1,\\0,&\,\, \mbox{otherwise}\nonumber
		\end{cases}
		\end{align}
		Then, the gradient of $J(\mathbf{W})$ with respect to $\mathbf{w}$ may be written as
		\begin{align}
		\bar{\pmb\rho}^\top\vecc\frac{\partial J(\mathbf{W})}{\partial\mathbf{W}}=\mathbf{0}\nonumber
		\end{align}
		where $\bar{\pmb\rho}=[\pmb\rho_1\,\pmb\rho_2\,\cdots\pmb\rho_K]\in\mathbb{R}^{N^2K\times K}$.
		Then, from above we have that
		\begin{align}
		\frac{\partial J(\mathbf{W})}{\partial\mathbf{w}}&=\bar{\pmb\rho}^\top\vecc\frac{\partial J(\mathbf{W})}{\partial\mathbf{W}}=\bar{\pmb\rho}^\top\left(\mathbf{F}\vecc\mathbf{W}+\mathbf{g}\right)=\mathbf{0}\nonumber
		\end{align}
		Since $\vecc\mathbf{W}=[\mathbf{w}^\top \underbrace{0\cdots0}_{NK}\mathbf{w}^\top\underbrace{0\cdots0}_{NK}\cdots\mathbf{w}^\top]^\top$, the values of $\mathbf{w}$ depend only on those columns of $\bar{\mathbf{F}}$ and $\bar{\mathbf{g}}$ which correspond to the $N\times K$ nonzero components of $\vecc\mathbf{W}$. 
		The component sets of $\vecc\mathbf{W}$ where $\mathbf{w}$ is present are the following:
		\begin{align}
		&\Omega_1\triangleq[1:K] \mbox{(first K corresponding to the first occurence of $\mathbf{w}$)}\nonumber\\
		&\Omega_2\triangleq[(N+1)K+1:(N+1)K+K]\nonumber\\
		&\Omega_3\triangleq[2(N+1)K+1:2(N+1)K+K]\nonumber\\
		&\Omega_4\triangleq[3(N+1)K+1:3(N+1)K+K]\nonumber\\
		\vdots\nonumber\\
		&\Omega_N\triangleq\underbrace{[(N-1)(N+1)K+1:(N-1)(N+1)K+K]}_{\equiv\displaystyle [N^2K-K+1:N^2K]  \mbox{(the last K corresponding to the last occurence of $\mathbf{w}$)}}\nonumber
		\end{align}

	\end{proof}
	\begin{figure*}[t]
		\centering
		\subfigure[]{\includegraphics[width=2.2in]{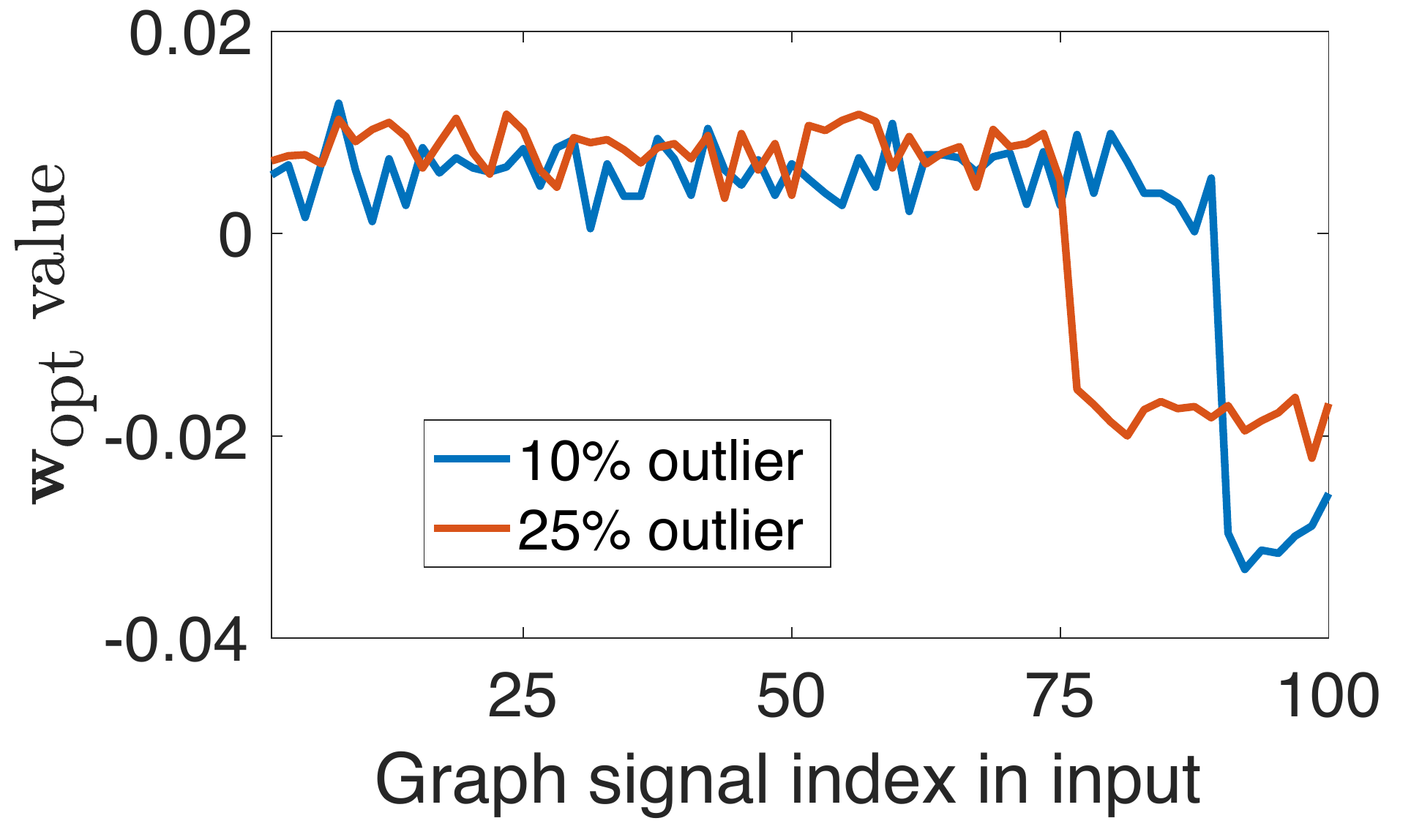}}
		\subfigure[]{\includegraphics[width=2.2in]{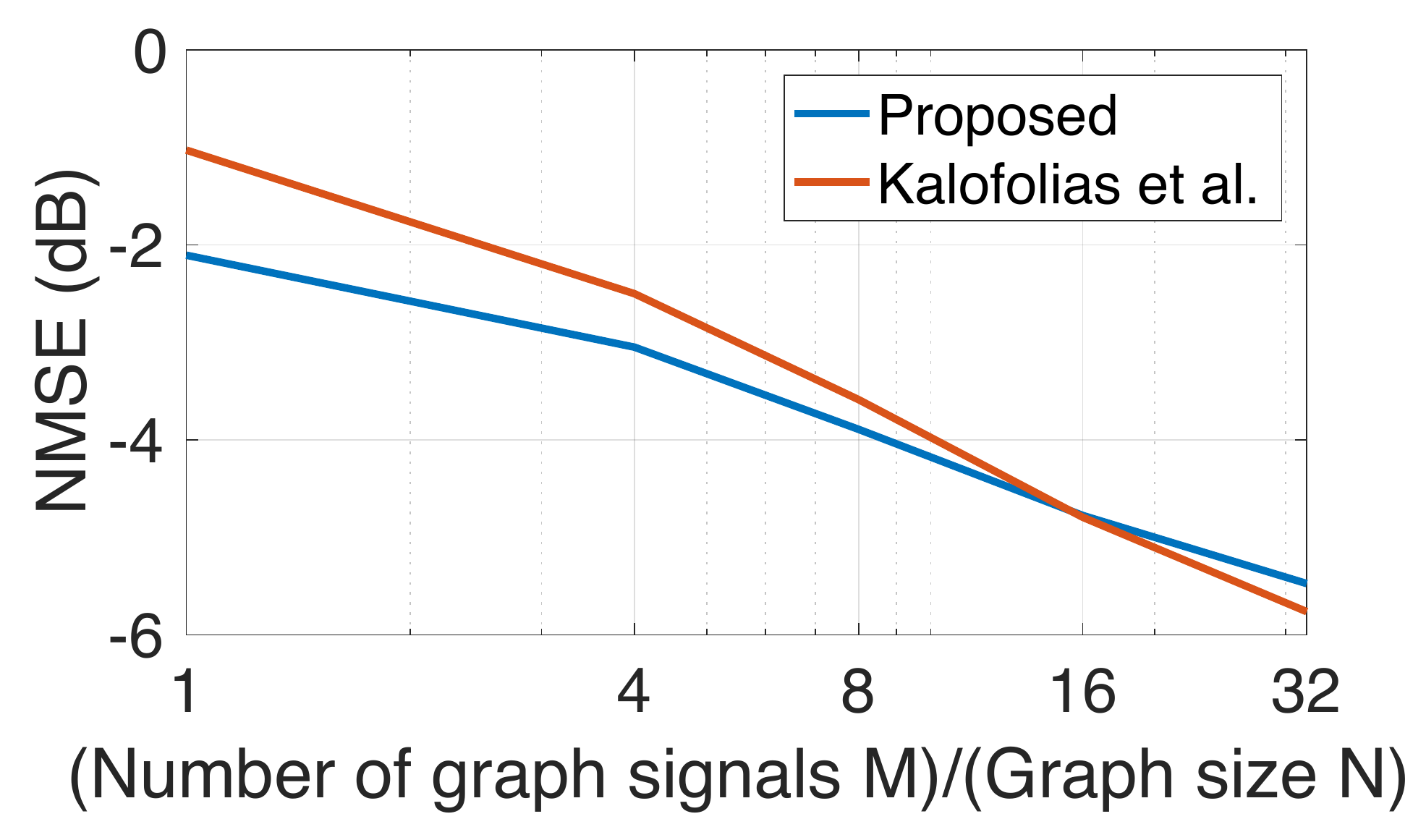}}
		\subfigure[]{\includegraphics[width=2.2in]{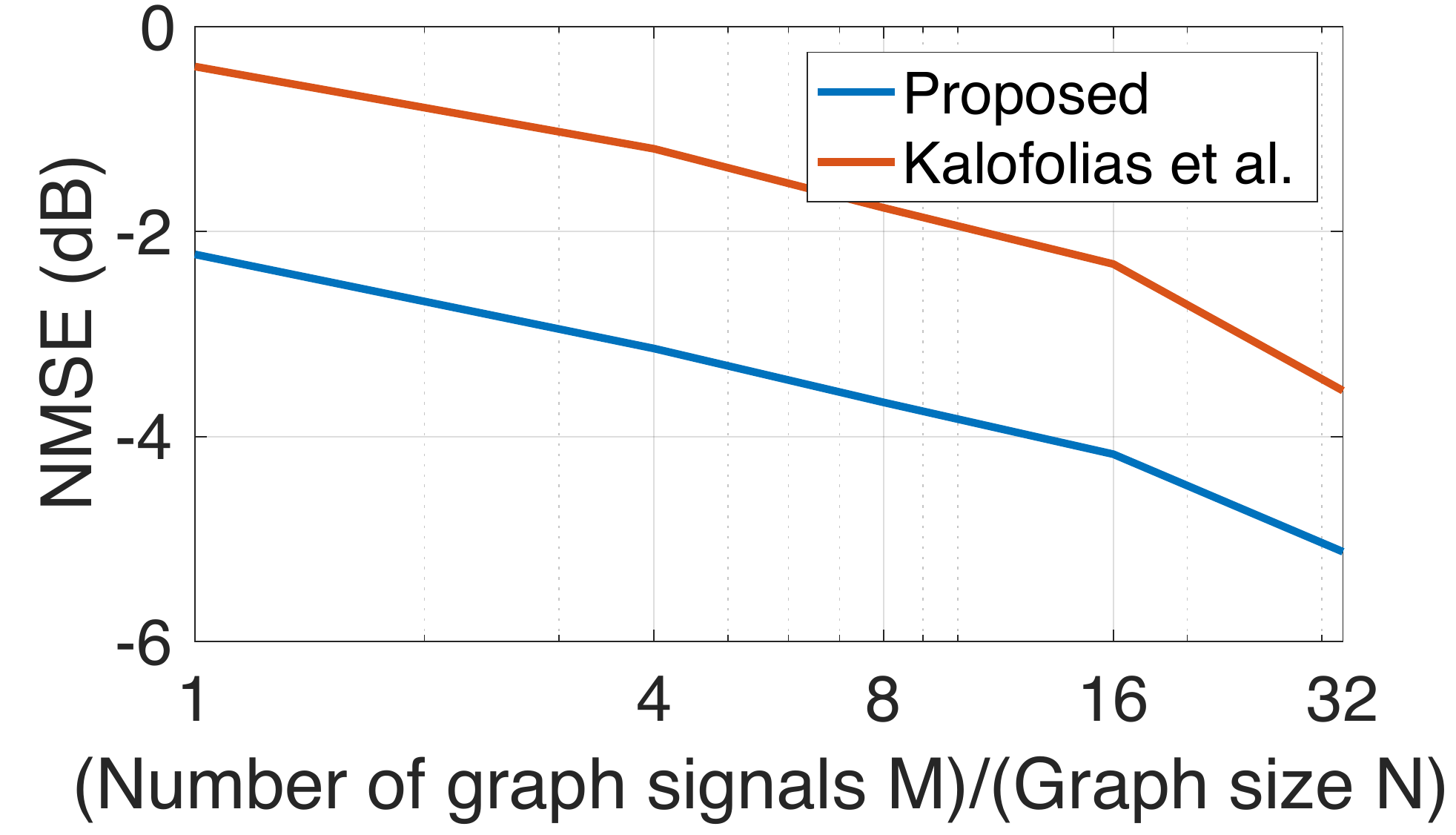}}
		\caption{Results from synthesized graph signal data with outliers (a) Plot showing $\mathbf{w}$ for sparse graph of size $N=6$ for $M=10N=100$ graph signals, (b) NMSE of both approaches with $10\%$ outliers, and (c) with $25\%$ outliers. }
		\label{fig:synth}
	\end{figure*}
	Then, from Proposition 1, the optimal coefficients $\mathbf{w}$ are obtained by solving:
	\begin{align}
	\bar{\mathbf{F}}\mathbf{w}_{\mbox{opt}}=\bar{\pmb\rho}^\top{\mathbf{g}}\nonumber
	\end{align}
	where $
	\bar{\mathbf{F}}\triangleq\left[
	\mathcal{C}_{\Omega_1}(\bar{\pmb\rho}^\top\mathbf{F})+\mathcal{C}_{\Omega_2}(\bar{\pmb\rho}^\top\mathbf{F})+\cdots+
	\mathcal{C}_{\Omega_N}(\bar{\pmb\rho}^\top\mathbf{F})
	\right].$
	Finally, we have that 
	\begin{align*}
	\mathbf{w}_{\mbox{opt}}=\bar{\mathbf{F}}^\dagger\bar{\pmb\rho}^\top{\mathbf{g}},\nonumber
	\end{align*}
	where $\dagger$ denotes the pseudo-inverse operation.
	
	Once the regression coefficients are obtained, the prediction of the edge weight between two nodes $i$ and $j$ for a possibly new graph $g_{new}$, given their corresponding $M$ graph signal values as input, is given by
	\begin{align}
{a_{i,j}^{(g_{new})}=\mathbf{w}_{\mbox{opt}}^\top\pmb\phi\left(\mathbf{x}^{(g_{new})}(i),\mathbf{x}^{(g_{new})}(j)\right)}.
	\label{eq:a_pred}
	\end{align}
	{\color{black}We also note here that since we model the edge-weights separately as Eq. \eqref{eq:glkr_main}, we do not make any assumptions regarding the size of the graphs in the test data, they could be of different number of nodes than that of the training data graphs. That is, the size of the graphs may be different in training and test datasets. Our approach remains applicable as long as the spectral profile of the graph signals of all graphs is assumed to be the same, as discussed earlier.}
	\section{Experimental results}
	
	We first consider the application of our approach to synthesized graph datasets. Our goal is to learn the optimal regression coefficients from training graphs to make predictions for the adjacency matrices of test graphs. We consider the case when the input observations are smooth graph signals but a fraction of them are outlier signals which are high-frequency graph signals. We assume the outliers to occur at the same observation indices among the $M$ graph signals in the input, in all the training and test graphs. Such an example simulates the case where the graph signals represent features which may be not equally relevant to graph learning process. We then use $\mathbf{w}_{\mbox{opt}}$ to make predictions for the graphs in the test set using Eq. \eqref{eq:a_est}. Our expectation is that values of $\mathbf{w}_{\mbox{opt}}$ will exhibit different trends for the smooth graph signals and the outliers signals.
	
	We consider $32\,(G=16)$ graphs of size $N=10$, out of which $16$ are used for training and the remaining for testing. We first construct a sparse connected graph $\mathbf{A}_0$ with only 40\% of the total number of non-diagonal entries being non-zero with values drawn from uniform random distribution over $[0,1]$. The training and test graphs $\mathbf{A}^{(g)}$ are then generated by randomly inserting values at 10\% of non-diagonal entries of $\mathbf{A}_0$, with values again drawn from the uniform random distribution over $[0,1]$. The adjacency matrices are scaled so that they all have unit Frobenius norm. For each of these graphs, we generate smooth graph signals by drawing realizations from zero-mean Gaussian with covariance matrix $[\mathbf{L}^{(g)}]^\dagger$. The outlier signals are generated from zero-mean multivariate Gaussian distribution with covariance matrix $[\mathbf{L}^{(g)}]^2$, which correspond to high-frequency graph signals(following discussion in Section \ref{sec:gsp}). We consider two cases: outliers being $10\%$ and $25\%$ of the $M$ input signals for different values of $M$. We use $\alpha=0.1/ M$ and $\beta=10/M$ which are values set by crossvalidation. We measure the performance of our approach in terms of the normalized mean square error  (NMSE) of the prediction for test data defined as \begin{align}
	\mbox{NMSE}=\frac{\mathbb{E}(\|\mathbf{A}-\hat{\mathbf{A}}\|_F^2)}
	{\mathbb{E}(\|\mathbf{A}\|_F^2)}\nonumber
	\end{align}
	where $\mathbf{A}$ and $\hat{\mathbf{A}}$ denote the true and estimated adjacency matrices, respectively; and the expectation is taken over all testing matrices for 100 Monte Carlo runs. In Figure \ref{fig:synth}(a), we show an instance of $\mathbf{w}$ for $M=10$, averaged over 100 Monte Carlo realizations for both $10\%$ and $25\%$ outlier cases. We observe that the computed regression coefficients clearly shows a difference in trend between smooth graphs signals and outliers. This is because we impose the resulting model to minimize the graph smoothness regularization. 
	In Figure \ref{fig:synth}(c), we plot the NMSE for test prediction at different values of input observation size $M$.
	{\color{black}Since no prior supervised learning approaches exist for the problem, we restrict ourselves to the comparison of our method with that of Kalofolias et al \cite{kalofolias2016learn}. The approach is a popular approach used in learning graphs from smooth graph signals. The hyperparameters of the method are set by choosing the values which minimize the training NMSE. In the case of both approaches, we threshold the estimated matrices in order to obtain sparse matrices.}
	
	We notice that, with $10 \%$ outliers in the input, our approach outperforms graph-learning when $M$ is small, but nearly coincides with the latter at larger $M$ values. When the fraction of outliers in the input is increased to $25\%$, we observe that the other approach performs poorly in comparison with our approach. Further, our approach results in an NMSE which is comparable with that obtained in the $10\%$ outlier case. Such a consistency in prediction NMSE of our approach may be attributed to the regression model which allows for differential treatment of the various input graph signals. The classical graph-learning on the other hand explicitly assumes all the signals to be smooth over the graph that is being inferred. 
	We note that in both $10\%$ and $25\%$ outlier cases, the NMSE of both approaches decreases as $M$ increases. In Table \ref{tab:F}, we list the F values of the obtained matrices for test data from both methods. We observe that the F scores of both approaches increase with $M$ and are close to each other.
	
	We further note here that we also performed the experiment with Erdos-Renyi graphs of size $N=10$, and observed similar trends in NMSE both with $M$ and in comparison with $M$. However, we omit the corresponding NMSE plots here for brevity.
	\begin{table}[t]
		\centering
		\begin{tabular}{|c|c|c|c|c|c}
			\hline
			$M/N$&Proposed&GL&Proposed&GL\\
			&10\%&10\%&25\%&25\%\\
			\hline
			1 & 0.71 & 0.66 & 0.71 & 0.62\\
			4& 0.72 & 0.74 & 0.74 & 0.69\\
			8 & 0.74 & 0.78 & 0.74 & 0.72\\ 
			16 & 0.74 & 0.82 & 0.73 & 0.76\\ 
			32 & 0.79 & 0.85 & 0.8 & 0.84 \\
			\hline
		\end{tabular}
		\caption{F values for estimated graphs for the two approaches}
		\label{tab:F}
	\end{table}
	\section{Conclusions}
	We proposed a supervised learning approach based on linear regression for predicting graphs from graph signals based. We formulated a convex optimization problem to solve for the regression coefficients. Our approach was shown to result in a good prediction performance when not all the graph signals may be equally relevant or may have noise/ corruptions. This is because the linear regression model allows for the different graph signals be weighed differently. The learnt regression coefficients were seen to reflect the presence of outliers in the graph signals. In future work, we will pursue the application of our approach to real-world datasets such as weather measurements, and Yelp, and functional magnetic resonance imaging data.
	
	\bibliographystyle{IEEEtran}
	\bibliography{bibliography.bib}

\end{document}